\documentclass[10pt,conference]{IEEEtran}

\usepackage{url,amssymb,amsmath,boldmath}

\usepackage{cite}

\newtheorem{theorem}{Theorem}
\newtheorem{corollary}[theorem]{Corollary}
\newtheorem{algorithm}[theorem]{Algorithm}

\def\F{{\mathbb{F}}}
\def\wgt{\mathop{\rm wgt}}
\def\MAGMA{{\sf MAGMA}}

\begin{document}
\pagestyle{plain}

\title{Computing Extensions of Linear Codes}

\author{\authorblockN{Markus Grassl}
\authorblockA{
Institut f\"ur Algorithmen und Kognitive Systeme\\
Fakult\"at f\"ur Informatik, Universit\"at Karlsruhe (TH)\\
Am Fasanengarten 5, 76\,128 Karlsruhe, Germany\\
Email: grassl@ira.uka.de}
}

%

\maketitle

\begin{abstract}
This paper deals with the problem of increasing the minimum distance
of a linear code by adding one or more columns to the generator
matrix.  Several methods to compute extensions of linear codes are
presented.  Many codes improving the previously known lower bounds on
the minimum distance have been found.
\end{abstract}

\section{Introduction}
In this paper we consider the question when a linear code
$C=[n,k,d]_q$ over $\F_q$ of length $n$, dimension $k$, and minimum
distance $d$ can be extended to a code $C'=[n+1,k,d+1]_q$.  It is a
well known fact in coding theory that every binary linear code
$C=[n,k,d]_2$ whose minimum weight $d$ is odd can be extended to a
code $[n+1,k,d+1]_2$ by adding a single parity check.  This can also
be expressed in terms of Construction~X \cite{SRC72} applied to the
code $C$, its one-codimensional even-weight subcode $C_0$, and the
trivial code $[1,1,1]_2$.  While this result does not have an
immediate generalization to non-binary alphabets, Hill and Lizak
\cite{HiLi95,Hi99} proved the following theorem:
\begin{theorem}
Let $C$ be an $[n,k,d]_q$ code with $\gcd(d,q)=1$ and with all
weights congruent to $0$ or $d$ (modulo $q$). Then $C$ can be extended
to an $[n+1,k,d+1]_q$ code all of whose weights are congruent to $0$
or $d+1$ (modulo $q$).
\end{theorem}
In order to apply this theorem, knowledge about the weight spectrum of
the code $C$ is required.  A generalization of this theorem due to
Simonis \cite{Sim00} can be applied when additionally information on
the weight distribution of the code $C$ is available.  The special
cases with $\gcd(q,d)=1$ and in particular ternary codes have been
treated by Maruta \cite{Mar03,Mar04,Mar05}.  However, these results
are of rather theoretical nature and have mainly be used to prove the
non-existence of codes with certain parameters.  The application to a
specific code might be difficult since one has to compute information
on the weight distribution of the code first.

\section{Extension Based on Minimum Weight Codewords}
\subsection{The main criterion}
In the following, we consider the problem to test if a code
$C=[n,k,d]_q$ which is explicitly given by a generator matrix $G$ can
be extended and to compute an extension if it exists.  Based on the
set of all codewords of minimum weight, we get the following criterion
for the extendability of a linear code:
\begin{theorem}\label{theorem:extension}
Let $C=[n,k,d]_q$ be a linear code over $\F_q$ with minimum distance
$d$.  Furthermore, let $G\in\F_q^{k\times n}$ be a generator matrix
for $C$ of full rank. By ${\cal S}_d=\{ \bm{c}\in C|\wgt{\bm{c}}=d\}$
we denote the set of all codewords of minimum weight and by ${\cal
J}_d=\{\bm{v}\in\F_q^k|\wgt(\bm{v}G)=d\}$ we denote the corresponding
information vectors.

The code $C$ can be extended to a code $C'=[n+m,k,d+1]_q$ if and only
if there is a matrix $X\in\F_q^{k\times m}$ such that
\begin{equation}\label{eq:gen_ext}
\sum_{i=1}^k v_i X_i\ne 0\quad\text{for all $\bm{v}\in {\cal J}_d$,}
\end{equation}
where $X_i$ denotes the $i$-th row of the matrix $X$.
\end{theorem}
\begin{proof}
Let $G'=(G|X)$ be the matrix that is obtained by appending the matrix
$X$ to $G$.  Encoding an information vector $\bm{v}$ with the matrix
$G'$ we get
\[
\bm{c}'=\bm{v}G'=(\bm{v}G|\sum_{i=1}^k v_i X_i).
\]
The weight of a non-zero codeword $\bm{c}'$ is $d$ if and only if
$\wgt(\bm{v}G)=d$ and $\sum_{i=1}^k v_i X_i=0$.
\end{proof}
In particular we consider the extension by a single column:
\begin{corollary}
Using the notation of Theorem~\ref{theorem:extension}, a linear code
$C=[n,k,d]_q$ can be extended to a code $C'=[n+1,k,d+1]_q$ if and only
if there exists a column vector $\bm{x}\in\F_q^k$ such that
\begin{equation}\label{eq:single_ext}
\sum_{i=1}^k v_i x_i\ne 0\quad\text{for all $\bm{v}\in {\cal J}_d$.}
\end{equation}
\end{corollary}
In order to apply criterion (\ref{eq:gen_ext}) or
(\ref{eq:single_ext}), we have to compute the set ${\cal J}_d$ of
information vectors of all codewords of minimum weight.

\subsection{Computing the minimum weight codewords}
In the sequel we describe an algorithm to compute the minimum distance
of a code as well as all words of minimum weight.  The algorithm is
based on an algorithm by Zimmermann to compute the minimum distance
(see \cite{Zim96} and \cite[Algorithmus 1.3.6]{BFKWZ98}) which
improved an algorithm by Brouwer.  Together with some further
improvements, the  algorithm is implemented in the computer algebra
system \MAGMA{} (see \cite{Magma,Gra:Codes06}).

The main idea of the algorithm is to enumerate the codewords in such a
way that one does not only obtain an upper bound on the minimum
distance of the code via the minimum of the weight of the words that
have been encountered, but to establish lower bounds on the minimum
distance as well. For this, we are using a collection of systematic
generator matrices $G_j$ with corresponding information sets ${\cal
I }_j$.  Given an  ordered list $({\cal I}_1,{\cal
I}_2,\ldots,{\cal I}_\mu)$ of information sets, we define the {\em
relative rank} $r_j$ of ${\cal I}_j$ as 
$$
r_j:=k-\Bigl|{\cal I}_j \cap \bigcup_{l=1}^{j-1}{\cal I}_l\Bigr|,
$$ 
i.e., $r_j$ equals the number of positions in the information set
${\cal I}_j$ that are disjoint from all informations sets ${\cal I}_l$
with $l<j$.  If we now encode all words $\bm{i}\in\F_q^k$ of weight
$\wgt{\bm{i}}\le w$ using all systematic generator matrices, we know
that the weight of the remaining codewords is at least
$$
d_{\text{lb}}:=\sum_{j=1}^\mu \max(0,(w+1)-(k-r_j)),
$$ 
as the weight in each corresponding information set is at least $w+1$,
and we have to subtract the positions which have been double-counted
due to overlapping information sets.

\begin{algorithm}[Minimum Weight Words]\label{algo:min_wgt_words}
\smallskip

\begin{small}
\begin{algo}
MinimumWeightWords:=function(C);
\  minwords:=$\emptyset$;
\  $d_{\text{lb}}$:=$1$;
\  $d_{\text{ub}}$:=$n-k+1$;
\  $w$:=$1$;
\  while $w\le k$ and $d_{\text{lb}}\le d_{\text{ub}}$ do
\    for $j$:=$1$ to $\mu$ do
\      words:=$\{\bm{i}\cdot G_j\colon\bm{i}\in\F_q^k|\wgt(\bm{i})=w\}$;
\      d:=$\min\{\wgt(\bm{c})\colon\bm{c}\in\text{words}\}$;
\      if d<$d_{\text{ub}}$ then
\        $d_{\text{ub}}$:=d;
\        minwords:=$\{\bm{c}\colon\bm{c}\in\text{words}|\wgt(\bm{c})=d_{\text{ub}}\}$;
\      else
\        minwords join:=$\{\bm{c}\colon\bm{c}\in\text{words}|\wgt(\bm{c})=d_{\text{ub}}\}$;
\      end if;
\    end for;
\    $d_{\text{lb}}$:=$\sum\limits_{j=1}^\mu \max(0,(w+1)-(k-r_j))$;
\    $w$:=$w+1$;
\  end while;
\  return minwords;
end function;
\end{algo}
\end{small}

\smallskip

\end{algorithm}

With a slight modification, this algorithm can also be used to compute
all codewords of a given weight or all codewords whose weight is below
a certain value. The total number of encodings to find all codewords
of weight not exceeding $d$ is given by
\begin{equation}\label{eq:work_factor_disjoint}
\sum_{w=1}^{w_0} \mu \binom{k}{w} (q-1)^{w-1},
\end{equation}
where $w_0$ is the minimum value such that
\begin{equation}\label{eq:lower_bound}
\sum_{j=1}^\mu \max(0,(w_0+1)-(k-r_j))>d.
\end{equation}
Of course, if (\ref{eq:work_factor_disjoint}) is larger than $q^k$,
one should directly enumerate all codewords instead of using
Algorithm~\ref{algo:min_wgt_words}.  But in most cases, using more
than one generator matrix results in an overall saving as the maximum
weight $w_0$ of the vectors $\bm{i}$ that has to be considered is
smaller, and (\ref{eq:work_factor_disjoint}) grows only linear in $\mu$,
but exponential in $w_0$.  If partial knowledge of the automorphism
group of the code is available, which is e.\,g. the case for cyclic or
quasi-cyclic codes, the lower bound (\ref{eq:lower_bound}) on $w_o$
can be improved so that the overall complexity for computing the
minimum weight codewords is reduced \cite{Che70,Whi05,Gra:Codes06}.

\section{Computing Extensions}
\subsection{Exhaustive search}\label{sec:brute_force}
Given the set ${\cal J}_d$ of information vectors of the minimum
weight codewords, one can use an exhaustive search to find a column
vector $\bm{x}$ or a matrix $X$ that fulfills condition
(\ref{eq:single_ext}) or (\ref{eq:gen_ext}).  In total there are
$q^{mk}-1$ non-zero matrices.  As both conditions are bilinear, it
suffices to consider normalized information vectors and may normalize
the columns in the matrix $X$, reducing the total number of matrices
by no more than the factor $(q-1)^m$.  Sorting the columns of the
matrix gives an additional reduction by a factor of at most $m!$.
Hence using this approach, one has to test at least
\begin{equation}\label{eq:brute_force}
\frac{q^{mk}-1}{m!(q-1)^m}
\end{equation}
matrices in order to show that no extension exists.  If one is
interested in all possible extension, an exhaustive search is
necessary, too.  Nonetheless, exhaustive search might be feasible to
find an extension if the dimension $k$ of the code is small or if many
extensions exist.

\subsection{Extending binary codes by one position}
For binary codes, condition (\ref{eq:single_ext}) can be re-written as
\begin{equation}\label{eq:lin_ext}
\sum_{i=1}^k v_i x_i=1\quad\text{for all $\bm{v}\in {\cal J}_d$.}
\end{equation}
The possible extensions of the code correspond to the set of solutions
of the inhomogeneous system of linear equations (\ref{eq:lin_ext}).
The complexity of computing the solutions if one exists is no longer
exponential as in (\ref{eq:brute_force}), but only polynomial.
Moreover, it suffices to compute a subset ${\cal J}'_d$ of the
information vectors of the minimum weight codewords such that the
linear spans of ${\cal J}_d$ and ${\cal J}'_d$ coincide.

\subsection{Extensions by one via solving polynomial equations}\label{sec:ideal1}
For non-binary codes, condition $(\ref{eq:single_ext})$ does not
directly translate into an equation.  However, using the fact that the
roots of the polynomial $y^{q-1}-1 \in\F_q[y]$ are exactly the
non-zero elements of $\F_q$, we get the condition
\begin{equation}\label{eq:poly_ext}
\left(\sum_{i=1}^k v_i x_i\right)^{q-1}=1\quad\text{for all $\bm{v}\in {\cal J}_d$.}
\end{equation}
The set of all solutions of conditions (\ref{eq:poly_ext}) is
characterized by the ideal
\begin{equation}\label{eq:ext_ideal}
J:=\left\langle \left(\sum_{i=1}^k v_i x_i\right)^{q-1}-1\colon\bm{v}\in {\cal J}_d\right\rangle
  \unlhd \F_q[x_1,\ldots,x_k]
\end{equation}
in the polynomial ring $\F_q[x_1,\ldots,x_k]$ in $k$ variables over
$\F_q$. Testing whether the system of polynomial equations
(\ref{eq:poly_ext}) has a solution and computing the solutions can be
done e.\,g. using Gr\"obner bases \cite{CLOS92}.  The system does not
have a solution if and only if a Gr\"obner basis of the ideal $J$
contains a non-zero constant polynomial.  In general, it is difficult
to estimate the complexity of computing a particular Gr\"obner basis,
and the complexity might be exponential.  However, computing a
Gr\"obner basis without homogenization quite often quickly shows that
there is no solution.  Using the algorithm $F_4$ of Faug\'ere to
compute a Gr\"obner basis \cite{Fau99} as implemented in the computer
algebra system \MAGMA{} \cite{Magma}, it was quite often faster to
compute all solutions via the Gr\"obner basis than finding a single
solution using exhaustive search (see below).

\subsection{General extensions via solving polynomial equations}
For both binary and non-binary codes, condition (\ref{eq:gen_ext}) can
be expressed in terms of polynomial equations.  A vector
$\bm{y}\in\F_q^m$ is non-zero if and only if at least one coordinate
is non-zero, i.\,e.
\[
\prod_{j=1}^m\left(y_j^{q-1}-1\right)=0.
\]
Hence the solutions of (\ref{eq:gen_ext}) are characterized by the
ideal
\begin{equation}\label{eq:ext_ideal_gen}
J:=\left\langle \prod_{j=1}^m\left(\left(\sum_{i=1}^k v_i X_{ij}\right)^{q-1}-1\right)
  \colon\bm{v}\in {\cal J}_d\right\rangle
\end{equation}
in the polynomial ring $\F_q[X_{11},\ldots,X_{km}]$ in $km$ variables
over $\F_q$.  Note that even for $q=2$, the conditions are no longer
linear, but of degree $m$.

\subsection{Further remarks}
For linear binary codes we have seen that sometimes it is sufficient
to compute only a subset of the minimum weight codewords.  In general,
one can use a subset of ${\cal J}_d$ to test whether a code can be
extended and compute a set of candidates for the extension using the
ideal $J$ of eq.~(\ref{eq:ext_ideal}) or eq.~(\ref{eq:ext_ideal_gen}).
In many cases, the resulting set of candidates is rather small, so
that one can perform an exhaustive search among them.  Similar, a
double extension of a code $C$ to a code $C''=[n+2,k,d+2]_q$ can be
found using the solutions for the single extension to
$C'=[n+1,k,d+1]_q$.

Kohnert \cite{Koh07,Koh06} has proposed to compute extensions using
integer linear programming by reformulating (\ref{eq:gen_ext}) as
hitting-set problem.  The ground set of the hitting-set problem is the
set of all normalized non-zero vectors that can be appended to the
generator matrix, so its size grows exponentially in the dimension of
the code.

\section{Examples}
We tested the various methods using the best known linear codes (BKLC)
from \MAGMA{} and the linear codes from \cite{codetables} which
establish or improve the lower bound on the minimum distance in
Brouwer's tables \cite{Bro98}.  We have not found any binary code that
can be extended by one position, but many codes over $\F_q$ for
$q=3,4,5,7,8,9$.  In Table~\ref{table:ext_m=1} we list 71 of these
codes together with some timing information. The columns with headings
${\cal S}_d$ and $|{\cal S}_d|$ provide the time to compute all
minimum weight words and the number of minimum weight words.  In the
columns \emph{full iteration} and \emph{iteration} the time needed to
find all or just one solution by exhaustive search (see
Sect.~\ref{sec:brute_force}) is given for some of the codes.  The next
four columns provide information on the approach of
Sect.~\ref{sec:ideal1} solving a system of polynomial equations.  We
have used the additional equations $x_1^2-x_1$ which ensures that the
first component of the column vector $\bm{x}$ is either zero or one,
and $x_j^q-x_j$ for $j=2,\ldots,k$ as all entries of $\bm{x}$ are
elements of $\F_q$.  The total running time is dominated by the time
needed to compute the \emph{Gr\"obner} basis, the construction of the
\emph{equations} and computing the \emph{solutions} can be neglected
in most of the cases.  In the final column the total number of
solutions is given, where we have identified solutions that differ by
a non-zero scalar factor.

With some few exceptions, e.\,g., for the codes $[89,11,54]_5$,
$[93,11,57]_5$, $[76,8,53]_7$, $[45,8,30]_9$, computing all solutions
via a Gr\"obner basis is even faster than finding a single solution by
exhaustive search.

Table~\ref{table:ext_m=2} contains some binary and ternary codes whose
minimum distance can be increased by appending two columns to the
generator matrix.  For these codes, the Gr\"obner basis approach is
quite fast, but unfortunately, this is not always true.

There is a ternary code $C=[178,23,81]_3$ with 80 words of weight 81
that can be extended to a code $C'=[179,23,82]_3$.  Computing a
Gr\"obner basis took about 79 hours on an AMD Opteron 252 (clock speed
2.6 GHz), using about 16 GB of memory.  Using exhaustive search, a
solution was found in 189.730 seconds, while the projected total
running time for the complete exhaustive search is more than 250
hours.

Furthermore, there is a quasicyclic code $C=[140,19,73]_4$ with 840
words of weight $73$ that can be extended to a code
$C'=[142,19,74]$. Using exhaustive search, a solution was found after
4.36 hours 
on an AMD Opteron 250 (clock speed 2.4 GHz). After 35.75 days of CPU
time, 654 solutions have been found while the projected total running
time for the exhaustive search is $10^{11}$ years.  Computing a
Gr\"obner basis for the ideal $J$ of this code seems to be infeasible.

It turns out that the codes with parameters $[66,22,22]_3$,
$[67,23,22]_3$, $[78,11,47]_5$, $[51,6,37]_7$, and $[76,8,53]_7$ are
doubly extendible.  The codes $[172,17,70]_2$ and $[173,18,70]_2$ in
Table~\ref{table:ext_m=2} can be extended in two steps to codes
$[175,17,72]_2$ and $[176,18,72]_2$.  The codes $[119,7,75]_3$ and
$[85,9,51]_3$ can even be extended to codes $[123,7,78]_3$ and
$[89,9,54]_3$.

\begin{table*}
\caption{Codes $C=[n,k,d]_q$ that can be extend to codes $C'=[n+1,k,d+1]_q$.\label{table:ext_m=1}}
\vskip-2ex
\centerline{\def\arraystretch{1.2}%
\begin{tabular}{c|r|r|r|r|rrrr|r}
code                       &
computing ${\cal S}_d$     & 
$|{\cal S}_d|$ & 
full iteration        &
iteration             & 
  equations           & 
  Gr\"obner           & 
  solution            & 
total Gr\"obner       & 
\#solutions \\     
\hline
$[233,9,146]_3$ & 0.130         & 1410  & 0.150         & 0.020         & 0.020         & 0.040         & 0.000         & 0.080         & 1\\
$[86,10,49]_3$  & 0.020         & 1008  & 0.430         & 0.030         & 0.020         & 0.030         & 0.000         & 0.050         & 3\\
$[175,10,103]_3$& 0.100         & 352   & 0.440         & 0.010         & 0.010         & 0.010         & 0.000         & 0.020         & 1\\
$[87,11,49]_3$  & 0.030         & 3312  & 1.250         & 0.280         & 0.070         & 0.100         & 0.000         & 0.180         & 2\\
$[176,13,97]_3$ & 0.120         & 66    & 11.370        & 0.000         & 0.000         & 0.140         & 0.000         & 0.150         & 9\\
$[100,19,43]_3$ & 2.080         & 21140 & 10602.910     & 1920.910      & 0.980         & 4.260         & 0.000         & 5.310         & 1\\
$[102,19,44]_3$ & 2.540         & 14492 & 9893.640      & 1257.940      & 0.650         & 2.020         & 0.000         & 2.710         & 1\\
$[166,19,81]_3$ & 11.640        & 328   & 9686.610      & 189.260       & 0.020         & 0.130         & 0.010         & 0.160         & 9\\
$[104,20,44]_3$ & 4.010         & 15722 & --            & 3513.530      & 0.770         & 2.420         & 0.000         & 3.250         & 1\\
 $[66,22,22]_3$ & 0.690         & 90    & --            & 160.600       & 0.000         & 780.780       & 26.570        & 807.380       & 465\\
$[108,22,43]_3$ & 13.450        & 102   & --            & 1820.900      & 0.000         & 604.690       & 0.010         & 604.710       & 12\\
$[165,22,75]_3$ & 140.230       & 96    & --            & 116.360       & 0.010         & 793.850       & 0.260         & 794.130       & 92\\
 $[67,23,22]_3$ & 0.440         & 134   & --            & 8931.530      & 0.010         & 43.320        & 2.450         & 45.800        & 201\\
 $[97,23,37]_3$ & 24.150        & 746   & --            & 154487.840    & 0.040         & 0.130         & 0.000         & 0.170         & 1\\
 $[99,23,38]_3$ & 30.600        & 658   & --            & 160156.420    & 0.030         & 0.100         & 0.000         & 0.140         & 1\\
$[111,23,44]_3$ & 40.170        & 114   & --            & 119090.890    & 0.000         & 905.710       & 0.000         & 905.720       & 3\\
$[149,23,64]_3$ & 149.090       & 108   & --            & 11522.410     & 0.010         & 1147.910      & 0.010         & 1147.930      & 23\\
$[166,23,75]_3$ & 238.290       & 200   & --            & 35682.570     & 0.010         & 10.390        & 0.000         & 10.400        & 3\\
$[191,23,89]_3$ & 420.910       & 98    & --            & 1736.210      & 0.000         & 9430.260      & 0.690         & 9430.960      & 123\\
$[191,24,88]_3$ & 722.510       & 112   & --            & 3220.530      & 0.010         & 2265.550      & 0.010         & 2265.570      & 15\\
$[194,24,90]_3$ & 830.840       & 112   & --            & 32986.960     & 0.010         & 2262.690      & 0.000         & 2262.700      & 13\\
$[197,24,92]_3$ & 1050.770      & 110   & --            & 10215.69      & 0.010         & 2396.040      & 0.030         & 2396.090      & 32\\
$[194,25,89]_3$ & 2390.880      & 114   & --            & 5909.010      & 0.000         & 6391.800      & 0.060         & 6391.860      & 39\\
$[215,25,103]_3$& 5443.010      & 164   & --            & 150817.640    & 0.010         & 97.980        & 0.000         & 97.990        & 1\\
$[178,27,77]_3$ & 13427.570     & 126   & --            & 301975.480    & 0.010         & 127409.310    & 0.070         & 127409.400    & 41\\
$[127,28,49]_3$ & 7995.850      & 12440 & --            & --            & 0.920         & 1.880         & 0.000         & 2.860         & 4\\
\hline\hline
$[135,6,96]_4$  & 0.030         & 225   & 0.050         & 0.000         & 0.010         & 0.000         & 0.000         & 0.010         & 2\\
$[159,7,111]_4$ & 0.050         & 2604  & 0.180         & 0.020         & 0.030         & 0.070         & 0.000         & 0.130         & 1\\
$[241,7,174]_4$ & 0.130         & 804   & 0.170         & 0.010         & 0.010         & 0.020         & 0.000         & 0.030         & 1\\
$[190,8,130]_4$ & 0.200         & 4164  & 0.720         & 0.030         & 0.080         & 0.080         & 0.000         & 0.170         & 3\\
$[191,8,130]_4$ & 0.130         & 4158  & 0.720         & 0.030         & 0.080         & 0.080         & 0.000         & 0.170         & 4\\
$[132,11,81]_4$ & 0.160         & 777   & 43.430        & 1.140         & 0.020         & 0.040         & 0.000         & 0.070         & 1\\
$[94,13,53]_4$  & 0.420         & 16890 & 738.800       & 1.760         & 0.540         & 1.270         & 0.000         & 1.870         & 3\\
$[129,13,77]_4$ & 0.660         & 15312 & 865.230       & 221.230       & 0.540         & 1.060         & 0.000         & 1.670         & 2\\
$[132,13,79]_4$ & 0.680         & 17136 & 747.580       & 217.610       & 0.630         & 1.280         & 0.000         & 1.980         & 3\\
$[142,13,85]_4$ & 0.720         & 8049  & 737.340       & 33.000        & 0.270         & 0.500         & 0.000         & 0.810         & 1\\
$[149,13,90]_4$ & 1.140         & 18318 & 764.780       & 21.600        & 0.660         & 2.020         & 0.010         & 2.750         & 4\\
$[161,13,98]_4$ & 1.340         & 31884 & 817.810       & 229.780       & 1.280         & 4.870         & 0.000         & 6.300         & 2\\
$[196,13,122]_4$& 1.380         & 168   & 745.220       & 14.860        & 0.000         & 43.960        & 0.000         & 43.970        & 3\\
$[120,14,69]_4$ & 1.920         & 315   & 2989.640      & 0.010         & 0.010         & 0.880         & 13.190        & 14.150        & 729\\
$[182,14,110]_4$& 3.530         & 19698 & 3142.570      & 741.160       & 0.760         & 2.430         & 0.000         & 3.280         & 6\\
$[134,15,77]_4$ & 5.220         & 50793 & 12051.040     & 463.570       & 2.110         & 13.010        & 0.000         & 15.350        & 4\\
$[183,15,110]_4$& 13.190        & 49218 & 12193.490     & 3525.470      & 2.320         & 18.940        & 0.010         & 21.500        & 3\\
$[45,16,17]_4$  & 0.180         & 192   & 47480.220     & 2358.470      & 0.010         & 2833.000      & 0.000         & 2833.010      & 3\\
$[91,16,47]_4$  & 6.430         & 3330  & --            & 1831.150      & 0.120         & 0.180         & 0.000         & 0.300         & 1\\
$[136,16,75]_4$ & 18.300        & 38880 & --            & 308.820       & 1.580         & 6.940         & 0.000         & 8.700         & 18\\
$[176,16,103]_4$& 29.980        & 219   & --            & 2779.760      & 0.010         & 3747.880      & 0.000         & 3747.890      & 1\\
$[64,17,29]_4$  & 4.430         & 6048  & --            & 699.140       & 0.220         & 0.340         & 0.000         & 0.580         & 3\\
$[116,17,61]_4$ & 25.710        & 249   & --            & 1.660         & 0.010         & 8275.260      & 0.320         & 8275.610      & 243\\
$[137,17,75]_4$ & 37.240        & 122751& --            & 2731.850      & 5.870         & 107.510       & 0.000         & 113.990       & 3\\
$[172,17,99]_4$ & 83.180        & 65325 & --            & 1379.040      & 3.230         & 30.070        & 0.000         & 33.670        & 27\\
$[87,19,41]_4$  & 125.400       & 2550  & --            & 126.980       & 0.130         & 0.520         & 0.000         & 0.660         & 4\\
$[95,19,45]_4$  & 50.430        & 11451 & --            & 43493.400     & 0.590         & 1.230         & 0.000         & 1.880         & 1\\
$[110,19,54]_4$ & 31.710        & 330   & --            & 177473.000    & 0.010         & 28449.570     & 0.000         & 28449.590     & 5\\
\end{tabular}%
}
\smallskip

Timings in seconds using Magma V2.13-8 on an AMD Opteron 252 (clock
speed 2.6 GHz, 16 GB RAM); for $q=3$, an AMD Opteron 254 (clock speed
2.8 GHz, 16 GB RAM) has been used.
\end{table*}

\begin{table*}\def\thetable{\ref{table:ext_m=1} (continued)}\addtocounter{table}{-1}
\caption{Codes $C=[n,k,d]_q$ that can be extend to codes $C'=[n+1,k,d+1]_q$.}
\vskip-2ex
\centerline{\def\arraystretch{1.2}%
\begin{tabular}{c|r|r|r|r|rrrr|r}
code                       &
computing ${\cal S}_d$     & 
$|{\cal S}_d|$ & 
full iteration        &
iteration             & 
  equations           & 
  Gr\"obner           & 
  solution            & 
total Gr\"obner       & 
\#solutions \\     
\hline
$[105,7,77]_5$  & 0.040         & 1760  & 1.010         & 0.070         & 0.050         & 0.230         & 0.000         & 0.280         & 1\\
$[78,11,47]_5$  & 0.060         & 780   & 682.690       & 53.860        & 0.060         & 55.730        & 0.000         & 55.790        & 3\\
$[84,11,51]_5$  & 0.220         & 3424  & 667.630       & 0.020         & 0.300         & 1.900         & 0.000         & 2.220         & 3\\
$[89,11,54]_5$  & 0.180         & 232   & 683.130       & 3.170         & 0.020         & 4324.910      & 0.020         & 4324.950      & 35\\
$[93,11,57]_5$  & 0.270         & 224   & 659.400       & 0.630         & 0.020         & 4172.490      & 0.060         & 4172.570      & 45\\
\hline
$[65,4,53]_7$   & 0.020         & 408   & 0.040         & 0.000         & 0.010         & 0.010         & 0.000         & 0.020         & 1\\
$[51,6,37]_7$   & 0.020         & 504   & 2.070         & 0.000         & 0.030         & 0.360         & 0.000         & 0.390         & 14\\
$[76,8,53]_7$   & 0.030         & 912   & 102.710       & 1.080         & 0.080         & 249.500       & 0.000         & 249.590       & 4\\
\hline
$[44,8,29]_8$   & 0.020         & 2443  & 376.920       & 2.830         & 0.090         & 2.240         & 0.000         & 2.340         & 1\\
$[68,8,49]_8$   & 0.100         & 12936 & 338.610       & 11.270        & 0.670         & 0.830         & 0.000         & 1.550         & 1\\
$[27,9,15]_8$   & 0.020         & 4914  & 2760.190      & 52.360        & 0.170         & 2.140         & 0.000         & 2.330         & 1\\
$[69,9,49]_8$   & 0.300         & 25480 & 2804.850      & 19.510        & 1.870         & 2.810         & 0.000         & 4.800         & 1\\
\hline
$[82,5,67]_9$   & 0.150         & 2176  & 1.700         & 0.030         & 0.050         & 0.200         & 0.000         & 0.260         & 1\\
$[87,6,69]_9$   & 0.400         & 4256  & 15.750        & 0.050         & 0.140         & 0.800         & 0.010         & 0.970         & 3\\
$[127,6,103]_9$ & 0.230         & 976   & 15.310        & 0.500         & 0.030         & 4.170         & 0.000         & 4.210         & 1\\
$[98,7,76]_9$   & 0.860         & 6776  & 146.120       & 1.430         & 0.360         & 2.550         & 0.000         & 2.950         & 1\\
$[45,8,30]_9$   & 0.220         & 1408  & 1332.960      & 325.270       & 0.070         & 10637.000     & 0.000         & 10637.080     & 1 
\end{tabular}%
}
\smallskip

\centerline{Timings in seconds using Magma V2.13-8 on an AMD Opteron
  252 (clock speed 2.6 GHz, 16 GB RAM).}
\end{table*}

\begin{table}[hbt]
\caption{Codes $C=[n,k,d]_q$ that can be extend to codes $C'=[n+2,k,d+1]_q$.\label{table:ext_m=2}}
\vskip-2ex
\centerline{\def\arraystretch{1.2}%
\begin{tabular}{c|r|r|r|r|r}
$C=[n,k,d]_q$                       &
 ${\cal S}_d$     & 
$|{\cal S}_d|$ & 
iteration      &     
Gr\"obner       & 
\#solutions \\
\hline
$[205,13,94]_2$ & 0.070 & 2169 & 152.280   & 0.440 & 3 \\
$[172,17,70]_2$ & 0.060 & 2616 & 71392.360 & 0.860 & 3 \\
$[166,18,66]_2$ & 0.050 & 1800 & --        & 0.600 & 3 \\
$[173,18,70]_2$ & 0.080 & 4230 & --        & 1.570 & 3 \\
$[205,19,82]_2$ & 0.080 & 1632 & --         & 0.600 & 9 \\
\hline
$[119,7,75]_3$  & 0.020 &  756 & 16.510   & 0.710 & 6 \\
$[85,9,51]_3$   & 0.030 & 4536 & 143.400  & 3.500 & 24\\
\end{tabular}%
}
\smallskip

{\footnotesize Timings in seconds using Magma V2.13-8 on an AMD Opteron 254, clock speed 2.8 GHz, 16 GB RAM.\par}
\end{table}

\section*{Acknowledgment}
The author would like to thank John Cannon, Allan Steel, and Greg
White from the \MAGMA{} group, University of Sydney, for their
support.
 
\IEEEtriggeratref{6}
\IEEEtriggercmd{\enlargethispage{-1.9in}}



%
%
%
%

\end{document}